\documentclass[runningheads]{llncs}
\usepackage{graphicx}
\usepackage{epsfig,epsf}
\usepackage{epstopdf}
\usepackage{graphics}
\usepackage{array}
\usepackage{amsmath}
\usepackage{amssymb}
\usepackage{comment}
\usepackage{algorithm}
\usepackage{algorithmic}
\newtheorem{obs}{Observation}
\usepackage{fullpage}
\date{}

\title{On the Complexity of Connected $(s,t)$-Vertex Separator}
\author{N.S.Narayanaswamy ${^1}$  \and N.Sadagopan ${^2}$} 
\institute{${^1}$ Department of Computer Science and Engineering, Indian Institute of Technology Madras, India \\ ${^2}$ Indian Institute of Information Technology, Design and Manufacturing, Kancheepuram, India. \\
\email{\{swamy@cse.iitm.ac.in, sadagopan@iiitdm.ac.in\}}}
\begin{document}
\maketitle
\begin{abstract}
We investigate the complexity of minimum connected $(s,t)$-vertex separator ($(s,t)$-CVS) and show that $(s,t)$-CVS is $\Omega(log^{2-\epsilon}n)$-hard for any $\epsilon >0$ unless NP has quasi-polynomial Las-Vegas algorithms.  i.e., for any $\epsilon >0$ and for some $\delta >0$, $(s,t)$-CVS is unlikely to have $\delta.log^{2-\epsilon}n$-approximation algorithm.  We then present an interesting chordality dichotomy: we show that $(s,t)$-CVS is NP-complete on graphs of chordality at least 5 and present a polynomial-time algorithm for $(s,t)$-CVS on chordality 4 graphs.  We also present a $\lceil\frac{c}{2}\rceil$-approximation algorithm for $(s,t)$-CVS on graphs with chordality $c$.  Finally, from the parameterized setting, we show that $(s,t)$-CVS parameterized above the $(s,t)$-vertex connectivity is $W[2]$-hard.
\end{abstract}
\section{Introduction}
The vertex or edge connectivity of a graph and the corresponding separators are of fundamental interest in Computer Science and Graph Theory.  Many kinds of vertex separators,  stable vertex separators \cite{stable}, clique vertex separators \cite{clique}, constrained vertex separators \cite{marx}, and $\alpha$-balanced separators \cite{marx} are of interest to the research community. 
As far as complexity results are concerned, finding a minimum vertex separator and a clique vertex separator are polynomial-time solvable, whereas, stable vertex separator and other constrained separators reported in \cite{marx} are NP-hard.  This shows that imposing an appropriate constraint on the well-studied vertex separator problem makes the problem NP-hard.  Interestingly, constrained vertex separators have received much attention in parameterized complexity as well \cite{marx,balanced}.   In particular, Marx in \cite{marx} considered the parameterized complexity of constrained separators satisfying some hereditary properties.   For example, clique separators and stable separators.  It is shown in \cite{marx} that the above problems have an algorithm whose running time is $f(k).n^{O(1)}$, where $k$ is the size of a constrained separator.  Algorithms of this nature are popularly known as fixed-parameter tractable algorithms with parameter as the solution size \cite{rolf}.   While many constrained vertex separators have attracted researchers from both classical and parameterized complexity, the related problem of finding a minimum connected $(s,t)$-vertex separator is open.  In light of \cite{marx}, this question can also be looked at as finding a $(s,t)$-vertex separator satisfying some non-hereditary property, like connectedness.  Moreover, the results in \cite{marx} do not carry over to connected $(s,t)$-vertex separator and the complexity of it remains open.   With these motivations, in this paper, we focus our attention on the computational complexity of minimum connected $(s,t)$-vertex separator ($(s,t)$-CVS).    \\
{\bf Remark:} The $(s,t)$-CVS can also be motivated from the theory of graph minors. We observe that there is an equivalence between the computational problems of finding a minimum connected $(s,t)$-vertex separator and a minimum set of edges whose contraction reduces the $(s,t)$-vertex connectivity to one.  It is important to note that the analogous computational problem of reducing the $(s,t)$-edge connectivity to zero by a minimum number of edge deletions is polynomial-time solvable, because this is computationally equivalent to finding a minimum $(s,t)$-cut and deleting all edges in it. \\
{\bf Our Results:} In this paper, we consider connected undirected unweighted simple graphs.   For a graph $G$, let $(s,t)$ denote a fixed non-adjacent pair of vertices in $G$.  Throughout this paper, when we refer to edge contraction, we do not contract edges incident on $s$ and edges incident on $t$.  \\
1. We establish a polynomial-time reduction from the Group Steiner Tree [ND12, \cite{garey}] to $(s,t)$-CVS.  Consequently, it follows that there is no polynomial-time approximation algorithm with approximation factor $\delta.log^{2-\epsilon}n$ for some $\delta >0$ and for any $\epsilon >0$, unless NP has quasi-polynomial Las-Vegas algorithms.\\
2. We then observe that on chordal graphs finding a minimum $(s,t)$-CVS is polynomial-time solvable as every minimal vertex separator is a clique.  We show that deciding $(s,t)$-CVS is NP-complete on chordality 5 graphs and on chordality 4 graphs $(s,t)$-CVS is polynomial-time solvable.  We also present a $\lceil\frac{c}{2}\rceil$-approximation algorithm for $(s,t)$-CVS on graphs with chordality $c$.  \\
3. We then consider designing algorithms for $(s,t)$-CVS whose running time is $f(k).n^{O(1)}$ where $k$ is the parameter of interest and $f(k)$ is a function independent of $n$.  If the parameter of interest is the chordality $c$ of the graph, then it follows from the above result that $(s,t)$-CVS is unlikely to have an algorithm whose running time is $f(c).n^{O(1)}, c \geq 5$, unless P=NP.   Whereas, on graphs of treewidth $\omega$, we show the existence of an algorithm for $(s,t)$-CVS with run time $f(\omega).n^{O(1)}$, here treewidth is the parameter of interest.  Algorithms with running time of this nature are well-studied in the literature and they are called fixed-parameter tractable algorithms in the theory of parameterized complexity \cite{rolf}.  Further, an important lower bound for $(s,t)$-CVS is the $(s,t)$-vertex connectivity itself.  It is now natural to consider the following parameterization: the size of a $(s,t)$-CVS minus the $(s,t)$-vertex connectivity.  This type of parameterization is known as above guarantee parameters \cite{aboveguarantee}.  We show that $(s,t)$-CVS parameterized above the $(s,t)$-vertex connectivity is unlikely to be fixed-parameter tractable under standard parameterized complexity assumption, and in the terminology of parameterized hardness theory, it is hard for the complexity class $W[2]$ in the $W$-hierarchy. \\
{\bf Graph Preliminaries:}
Notation and definitions are as per \cite{golu,west}.  Let $G =(V,E)$ be a connected undirected unweighted simple graph where $V(G)$ is the set of vertices and $E(G)$ is the set of edges.  For $S \subset V(G)$, $G[S]$ denote the graph induced on the set $S$ and $G \setminus S$ is the induced graph on the vertex set $V(G) \setminus S$. A vertex separator $S \subset V(G)$ is called a $(s,t)$-vertex separator if in $G \setminus S$, $s$ and $t$ are in two different connected components and $S$ is minimal if no proper subset of it is a $(s,t)$-vertex separator.  A minimum $(s,t)$-vertex separator is a minimal $(s,t)$-vertex separator of least size.  The $(s,t)$-vertex connectivity denote the size of a minimum $(s,t)$-vertex separator.  A connected $(s,t)$-vertex separator $S$ is a $(s,t)$-vertex separator such that $G[S]$ is connected and such a set $S$ of least size is a minimum connected $(s,t)$-vertex separator.   For a minimal $(s,t)$-vertex separator $S$, let $C_s$ and $C_t$ denote the connected components of $G \setminus S$ such that $s$ is in $C_s$ and $t$ is in $C_t$.  We let $G \cdot e$ denote the graph obtained by contracting the edge $e=\{u,v\}$ in $G$ such that $V(G \cdot e)=V(G)\setminus\{u,v\}\cup\{z_{uv}\}$ and $E(G \cdot e)=\{ \{z_{uv},x\} ~|~ \{u,x\}$ or $\{v,x\} \in E(G)\} \cup \{\{x,y\} ~|~ \{x,y\} \in E(G) \mbox{ and }x \not= u, y \not= v\}$.  A graph is said to have chordality $c$, if it contains no induced cycle of length at least $c+1$.  i.e. every cycle $C$ of length at least $c+1$ in $G$ has a chord (an edge joining a pair of non consecutive vertices in $C$). An optimization problem $\cal{P}$ is $\Omega(f(n))$-hard if there exists a constant $c > 0$ so that $\cal{P}$ admits no $c.f(n)$-approximation algorithm, unless P=NP (or NP has quasi-polynomial Las-Vegas algorithms). \\ 
{\bf Roadmap:} In Section \ref{complexity}, we analyze the complexity of $(s,t)$-CVS and present various hardness results.  In Section \ref{algorithms}, we present an approximation algorithm and polynomial-time algorithms for $(s,t)$-CVS in special graph classes.
\section{Complexity of $(s,t)$-CVS: Classical and Parameterized Hardness}
\label{complexity}
We first establish a classical hardness of $(s,t)$-CVS by presenting a polynomial-time reduction from the Group Steiner tree to $(s,t)$-CVS.  Moreover, the same reduction establishes an hardness of approximation for $(s,t)$-CVS.  We then analyze $(s,t)$-CVS on graphs with chordality $c$ and show that it is NP-complete on chordality 5 graphs.  We conclude this section by showing that $(s,t)$-CVS parameterizing above the $(s,t)$-vertex connectivity is $W[2]$-hard. 
\subsection{Classical Hardness: A Reduction from Group Steiner tree to $(s,t)$-CVS}
The decision version of $(s,t)$-CVS is given below
\begin{center}
\begin{tabular}{|p{14cm}|}
\hline 
{\bf Instance:} A graph $G$, a non-adjacent pair $(s,t)$, and $q \in \mathbb{Z^+}$\\
{\bf Question:} Is there a $(s,t)$-vertex separator $S \subset V(G)$, $|S| \leq q$ and $G[S]$ is connected? \\
\hline 
\end{tabular}
\end{center} 
The Group Steiner tree problem can be stated as follows: given a connected undirected unweighted graph $G$, an integer $r$, and a collection of sets, which we call groups $g_1,g_2,\ldots,g_l \subseteq V(G)$,  the objective is to find a subtree $T$ of $G$ with at most $r$ edges that contains at least one vertex from each group $g_i$.   We assume that the groups are disjoint.  The Group Steiner tree problem is a generalization of the Steiner tree problem \cite{garey} and therefore, it is NP-complete.   \\
We transform an instance $I=(G, g_1,g_2,\ldots,g_l \subseteq V(G), r)$ of the Group Steiner tree to the corresponding instance $I'=(G',s,t,l+r+1)$ of $(s,t)$-CVS as follows:  $V(G')=V(G) \cup \{s,t\} \cup \{x_i ~|~ 1 \leq i \leq l\}$. $E(G')=E(G) \cup \{\{s,x_i\} ~|~ 1 \leq i \leq l\} \cup \{\{t,x_i\} ~|~ 1 \leq i \leq l\} \cup \{\{x_i,y\} ~|~  y \in g_i \mbox{ and } 1 \leq i \leq l \}$.  An example is illustrated in Figure \ref{groupfig}.    
\begin{figure}
\begin{center}
\includegraphics[scale=0.7]{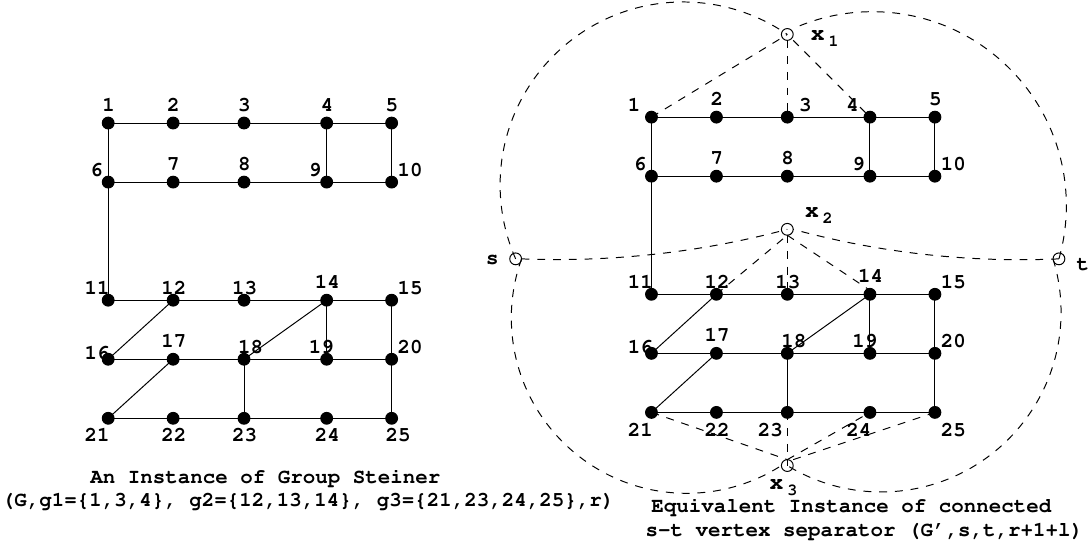}
\caption{An instance of Group Steiner tree reduces to an instance of $(s,t)$-CVS}\label{groupfig}
\end{center}
\end{figure}
\begin{theorem}
\label{approx-redn}
For $I$ and $I'$ as defined above,  $G$ has a Group Steiner tree with at most $r$ edges  if and only if $G'$ has a $(s,t)$-CVS of size at most $r+1+l$.
\end{theorem}
\begin{proof}
We first prove the necessity.  Given that $G$ has a Group Steiner tree $T$ with at most $r$ edges that contains at least one vertex from each group $g_i$.  By the construction of $G'$, it is clear that the $(s,t)$-vertex connectivity is $l$.  Therefore, any $(s,t)$-CVS in $G'$ has at least $l$ vertices.  Clearly, these $l$ new vertices together with at most $r+1$ vertices in $T$ form a $(s,t)$-CVS of size at most $r+1+l$ in $G'$.  Conversely, by the construction of $G'$,  any $(s,t)$-CVS $S$ of size at most $r+1+l$ must contain all $x_i$'s.  i.e. $N_{G'}(s) \subset S$.  This is true because $N_{G'}(s)$ is a $(s,t)$-vertex separator.  Since $S$ is connected and $N_{G'}(s)$ is an independent set, it follows that by the construction $S \setminus N_{G'}(s)$ is connected.   Moreover, $S$ must contain at least one element of $N_{G'}(x_i)$ for each $x_i$.   Since $|S \setminus N_{G'}(s)| \leq r+1$,  any spanning tree on $S \setminus N_{G'}(s)$ is a Group Steiner tree with at most $r$ edges.  Hence, the theorem follows. \qed
\end{proof}
As a consequence of the above theorem, it follows that $(s,t)$-CVS is NP-hard and it is easy to verify that $(s,t)$-CVS is in NP as certificate testing can be done in polynomial time using standard graph traversals \cite{coreman}.   Therefore, $(s,t)$-CVS is NP-complete.  We now show that our reduction establishes a stronger result: $(s,t)$-CVS is $\Omega(log^{2-\epsilon}n)$-hard, for all $\epsilon >0$ unless NP has quasi-polynomial Las-Vegas algorithms. \\
{\bf Hardness of Approximation of $(s,t)$-CVS:}  The Group Steiner tree problem with $l$ groups is at least as hard as the Set Cover problem, thus can not be approximated to a factor $o(\log l)$, unless $P=NP$ \cite{feige}.  On the hardness of approximation due to \cite{group-hard}, the following result is known:  Group Steiner tree problem is $\Omega(\log^{2-\epsilon}n)$-hard, for all $\epsilon >0$ unless NP has quasi-polynomial Las-Vegas algorithms.  i.e. there is no polynomial-time approximation algorithm for Group Steiner tree with approximation factor $\delta.log^{2-\epsilon}n$ for some $\delta >0$ and for any $\epsilon >0$, unless NP has quasi-polynomial Las-Vegas algorithms.   We now show that the above reduction is an approximation-ratio preserving reduction.  Let $OPT_g$ and $OPT_c$ denote the size of any optimum solution of the Group Steiner tree problem and the $(s,t)$-CVS problem, respectively.  Note that $OPT_c=OPT_g+l$ and $OPT_g \geq l$.  Suppose there is an $(1+\alpha)$-approximation algorithm for $(s,t)$-CVS, where $\alpha \leq \delta.log^{2-\epsilon}n$, for some $\delta,\epsilon >0$.  Then the size of the output of the algorithm is $(1+\alpha) OPT_c= (1+\alpha)(OPT_g+l) \leq (1+\alpha)(OPT_g+ OPT_g)  = 2(1+\alpha)OPT_g$. \\  This implies $2(1+\alpha)$-approximation algorithm for the Group Steiner tree problem, which is unlikely, unless NP has quasi-polynomial Las-Vegas algorithms \cite{group-hard}. 
\subsection{Complexity of $(s,t)$-CVS on Chordality $c$ Graphs}
Although $(s,t)$-CVS in general graphs is NP-complete, $(s,t)$-CVS is polynomial-time solvable on chordal graphs.  This is true because in chordal graphs, every minimal vertex separator is a clique \cite{golu}.  Now, this line of thought leaves open the complexity of $(s,t)$-CVS in chordality $c$ graphs.  A graph is said to have chordality $c$, if there is no induced cycle of length at least $c+1$.  Note that chordal graphs have chordality 3.  We now show that $(s,t)$-CVS on chordality 5 graphs is NP-complete.  
\begin{theorem}
\label{ch5-npc}
$(s,t)$-CVS is NP-complete on chordality 5 graphs.
\end{theorem}
\begin{proof}
{\bf $(s,t)$-CVS is in NP: } Given a certificate on an input instance $(G,s,t,q)$ of $(s,t)$-CVS, the certificate on Yes instances is a set $S \subseteq V(G)$ which is a connected $(s,t)$-vertex separator of cardinality at most $q$.  Clearly, $S$ can be verified in polynomial time by standard reachability algorithms \cite{coreman}. \\ 
{\bf $(s,t)$-CVS is NP-hard: }It is known from \cite{chordal-steinernpc} that Steiner tree problem on split graphs is NP-complete and this can be reduced in polynomial time to $(s,t)$-CVS in chordality 5 graphs using the following construction.   Note that any split graph $G$ can be seen as a graph with $V(G)=V_1 \cup V_2$ such that $G[V_1]$ is a clique and $G[V_2]$ is an independent set.  Also, split graphs are a subclass of chordal graphs and hence have chordality 3.   We map an instance $(G,R,r)$ of Steiner tree problem on split graphs to the corresponding instance $(G',s,t,r+1)$ of $(s,t)$-CVS as follows:  $V(G')=V(G) \cup \{s,t\}$ and $E(G')=E(G) \cup \{\{s,v\} ~|~ v \in R \} \cup \{ \{t,v\} ~|~ v \in R\}$. 
We now show that instances created by this transformation have chordality 5.  i.e., in $G'$, any cycle $C$ of length at least 6 has a chord.
Clearly,  $C$ must contain either $s$ or $t$ but not both.   Let $\{s,u_1,\ldots,u_p\}, p \geq 5$ denote the ordering of vertices in $C$. \\
{\bf Case 1:} $\{u_1,u_p\} \subseteq V_2$.  Since $G$ is a split graph, $\{u_2,u_{p-1}\} \subset V_1$, and therefore, $\{u_2,u_{p-1}\} \in E(G)$ which is a chord in $C$.  \\
{\bf Case 2:} $u_1 \in V_2$ and $u_p \in V_1$.   Clearly, $u_2 \in V_1$ and $\{u_2,u_p\} \in E(G)$, a chord in $C$.  \\
Therefore, we conclude that chordality of $G'$ is 5.   A proof in the similar line of Theorem \ref{approx-redn} argues that $G$ has a Steiner tree with at most $r$ edges if and only if $G'$ has a $(s,t)$-CVS with at most $r+1$ vertices.  As a consequence, it follows that $(s,t)$-CVS in chordality 5 graphs is NP-hard.  Thus, we conclude $(s,t)$-CVS in chordality 5 graphs is NP-complete.    \qed
\end{proof}
\subsection{$(s,t)$-CVS Parameterized above the $(s,t)$-vertex connectivity is $W[2]$-hard}
We consider the following parameterization which is the size of $(s,t)$-CVS minus the $(s,t)$-vertex connectivity.  Since the size of every $(s,t)$-CVS is at least the $(s,t)$-vertex connectivity, it is natural to parameterize above the $(s,t)$-vertex connectivity and its parameterized version is defined below.  
\begin{center}
\begin{tabular}{|p{12cm}|}
\hline
{\bf $(s,t)$-CVS parameterized above the $(s,t)$-vertex connectivity:}\\
{\bf Instance:} A graph $G$, a non-adjacent pair $(s,t)$ with $(s,t)$-vertex connectivity $k$ and $r \in \mathbb{Z^+}$ \\
{\bf Parameter:} $r$ \\
{\bf Question:} Is there a $(s,t)$-vertex separator $S \subset V(G)$, $|S| \leq k+r$ such that $G[S]$ is connected? \\
\hline 
\end{tabular}
\end{center}
We now show that there is no fixed-parameter tractable algorithm for $(s,t)$-CVS parameterized above the $(s,t)$-vertex connectivity.  In order to characterize those problems that do not seem to admit a fixed-parameter tractable algorithms, Downey and Fellows defined a {\em parameterized reduction} and a hierarchy of intractable parameterized problem classes above FPT, the popular classes are  $W[1]$ and $W[2]$.  We refer \cite{rolf} for details about parameterized reductions.   
We now present a parameterized reduction from parameterized Steiner tree problem to $(s,t)$-CVS parameterized above the $(s,t)$-vertex connectivity.  This parameterized version of Steiner tree problem is shown to be $W[2]$-hard in \cite{downey}.
\begin{center}
\begin{tabular}{|p{12cm}|}
\hline
{\bf Parameterized Steiner tree problem:} \\
{\bf Instance:} A graph $G$, a terminal set $R \subseteq V(G)$, and an integer $p$\\
{\bf Parameter:} $p$ \\
{\bf Question:} Is there a set of vertices $T \subseteq V(G) \setminus R$ such that $|T| \leq p$ and $G[R \cup T]$ is connected? $T$ is called Steiner set (Steiner vertices).\\
\hline 
\end{tabular} 
\end{center} 
\begin{theorem}
$(s,t)$-CVS Parameterized above the $(s,t)$-vertex connectivity is $W[2]$-hard.
\end{theorem}
\begin{proof}
Given an instance $(G,R,r)$ of Steiner tree problem, we construct the corresponding instance $(G',s,t,k,r)$ of $(s,t)$-CVS with the $(s,t)$-vertex connectivity $k=|R|$ as follows:  $V(G')=V(G) \cup \{s,t\}$ and $E(G')=E(G) \cup \{\{s,v\} ~|~ v \in R \} \cup \{ \{t,v\} ~|~ v \in R\}$.
We now show that $(G,R,r)$ has a Steiner tree with at most $r$ Steiner vertices if and only if $(G',(s,t),k,r)$ has a $(s,t)$-CVS of size at most $k+r$.  For {\em only if} claim,  $G$ has a Steiner tree $T$ containing all vertices of $R$ and at most $r$ Steiner vertices.  By our construction of $G'$, to disconnect $s$ and $t$, we must remove the set $N_{G'}(s)$ which is $R$, as there is an edge from each element of $N_{G'}(s)$ to $t$.  Since $G$ has a Steiner tree with at most $r$ Steiner vertices, implies that in $G'$, it guarantees a $(s,t)$-CVS of size at most $k+r$.  For {\em if} claim, $G'$ has a $(s,t)$-CVS $S$ with at most $k+r$ vertices.  Since the $(s,t)$-vertex connectivity is $k$ and $S$ is a $(s,t)$-vertex separator, from our construction of $G'$ it follows that $N_{G'}(s) \subseteq S$ and $k=|N_{G'}(s)|$.  This implies that $G$ has a Steiner tree with $R=N_{G'}(s)$ as the terminal set and $S \setminus N_{G'}(s)$ as the Steiner vertices of size at most $r$.  Hence the claim.  $|V(G')|=|V(G)|+2$ and $|E(G')| \leq |E(G)| + 2 |V(G)|$ and the construction of $G'$ takes $O(|E(G)|)$.  Clearly, the reduction is a parameter preserving parameterized reduction. Therefore, we conclude that deciding whether a graph has a $(s,t)$-CVS is $W[2]$-hard with parameter $r$.  \qed
\end{proof}
{\bf Remark:} The natural parameter for $(s,t)$-CVS is the size of $(s,t)$-CVS.  A recent result due to \cite{stcvs-fpt} shows that $(s,t)$-CVS parameterized by the size is fixed-parameter tractable.
\section{Algorithms for $(s,t)$-CVS on Graphs of Bounded Chordality}
\label{algorithms}
The focus of this section is to present a polynomial-time approximation algorithm with ratio $(\lceil\frac{c}{2}\rceil)$ for $(s,t)$-CVS on graphs with chordality $c$ and a polynomial-time algorithm for $(s,t)$-CVS on chordality 4 graphs.  Also, the existence of a  		polynomial-time algorithm for $(s,t)$-CVS on graphs of bounded treewidth is shown.
\begin{lemma}
\label{chordality-l-lemma}
Let $G$ be a graph of chordality $c$.  For each minimal vertex separator $S$, for each $u,v \in S$ such that $\{u,v\} \notin E(G)$, there exists a path of length at most $\lceil \frac{c}{2} \rceil$ whose internal vertices are in $C_s$ or $C_t$, where $C_s$ and $C_t$ are components in $G \setminus S$ containing $s$ and $t$, respectively. 
\end{lemma}
\begin{proof}
Suppose for some non-adjacent pair $\{u,v\} \subseteq S$, both $P^1_{uv}$ and $P^2_{uv}$ are of length more than $\lceil \frac{l}{2} \rceil$, where $P^1_{uv}$ and $P^2_{uv}$ are shortest paths from $u$ to $v$ whose internal vertices are in $C_s$ and $C_t$, respectively. 
Now, there is an induced cycle $C$ containing $u$ and $v$ such that $|C| > \lceil \frac{l}{2} \rceil$ + $\lceil \frac{l}{2} \rceil=l$.  However, this contradicts the fact that $G$ is of chordality $l$.   \qed
\end{proof}
\subsection{$(\lceil\frac{c}{2}\rceil)$-Approximation for $(s,t)$-CVS on Graphs with Chordality $c$}            
Let $OPT$ denote the size of any minimum $(s,t)$-CVS on chordality $c$ graphs.  Clearly, $OPT \geq k$, where $k$ is the $(s,t)$-vertex connectivity.  The description of approximation algorithm $ALG$ is as follows:
\begin{itemize}
\item[1.] Compute a minimum $(s,t)$-vertex separator $S$ in $G$. $S=\{v_1,\ldots,v_k\}$ be an arbitrary ordering of vertices in $S$.\\
\item[2.] For each non-adjacent pair $\{v_i,v_{i+1}\} \subseteq S, 1\leq i \leq k-1$, find a path $P_{v_iv_{i+1}}$ of length at most $\lceil \frac{c}{2} \rceil$ whose internal vertices are in $C_s$ or $C_t$.  Such a path exists as per Lemma \ref{chordality-l-lemma}. 
$S'={\displaystyle \bigcup_{1 \leq i \leq k-1} V(P_{v_iv_{i+1}})} \cup S$. 
\end{itemize}
Observe that $S'$ is a $(s,t)$-CVS in $G$.  The upper bound on the size of $S'$ output by $ALG$ is: $|S'| \leq k+(k-1)(\lceil \frac{c}{2} \rceil-1)$.  Therefore, approximation ratio $\beta$ is \\ \\
$\beta \leq \frac{k+(k-1)(\lceil \frac{c}{2} \rceil-1)}{k} = 1+(1-\frac{1}{k})(\lceil \frac{c}{2} \rceil-1) < 1+(\lceil \frac{c}{2} \rceil-1)=\lceil \frac{c}{2} \rceil$
\subsection{$(s,t)$-CVS in Chordality 4 Graphs is Polynomial time}
We have already mentioned that $(s,t)$-CVS in chordality 3 graphs is polynomial time as every minimal vertex separator is a clique and due to Theorem \ref{ch5-npc}, it is NP-complete on chordality 5 graphs.  This observation leaves open the question of $(s,t)$-CVS in chordality 4 graphs.  We now present a structural result about minimal vertex separators in chordality 4 graphs, using which we show that $(s,t)$-CVS in chordality 4 graphs is polynomial-time solvable.  This is the dichotomy we mentioned in the abstract.  
\begin{theorem}
\label{ch4-char-thm}
Every minimal $(s,t)$-vertex separator $S$ in a chordality 4 graph $G$ satisfies one of the following properties:
\begin{itemize}
\item[(1)] $G[S]$ is connected. 
\item[(2)] Let $\{X_1,\ldots,X_r\}, r \geq 2$ denote the set of connected components in $G[S]$ and $V(X_i)$ denotes the vertex set of the component $X_i$.  In $G \setminus S$, there exists $u$ in $C_s$ and there exists $v$ in $C_t$ such that for all $1 \leq i \leq r, N_G(u) \cap V(X_i) \neq \emptyset$ and $N_G(v) \cap V(X_i) \neq \emptyset$, where $C_s$ and $C_t$ denote the connected components in $G \setminus S$ containing $s$ and $t$, respectively.
\end{itemize}
\end{theorem}
\begin{proof}
Our proof is by induction on $n=|V(G)|$.  It is easy to verify that the statement of the lemma is true when $|V(G)| = 1$, $|V(G)|=2$, and $|V(G)|=3$.  Let us now assume all chordality 4 graphs on $n-1$ vertices satisfy our claim.  Consider a chordality 4 graph $G$ on $n$ vertices.  Consider a minimal $(s,t)$-vertex separator $S$ with $|S|=1$.  Then $S$ contains a cut vertex and our claim is true.  If $|S|=2$, then either $G[S]$ is connected or we can find $u$ in $C_s$ such that $S \subset N_G(u)$ and $v$ in $C_t$ such that $S \subset N_G(v)$.  Consider the case when $|S| \geq 3$.  Let $C_s$ and $C_t$ denote components in $G \setminus S$ containing $s$ and $t$, respectively.  Without loss of generality, we assume that both $C_s$ and $C_t$ contain at least two elements.  Otherwise, it must be the case that $S=N_G(s)$ or $S=N_G(t)$.
\\ {\bf Case 1:} $G[S]$ is not an independent set.  Let $e=\{x,y\}$ be an edge contained in a connected component $X$ of $G[S]$.  Consider the graph $G \cdot e$ obtained from $G$ by contracting $e$.  Clearly, $|V(G \cdot e)|=n-1$. Let $S'=S \setminus \{x,y\} \cup \{z_{xy}\}$.  Edges incident on $x$ or $y$ are now incident on $z_{xy}$.  Observe that $S'$ is a minimal $(s,t)$-vertex separator in $G \cdot e$.  If $G[S']$ is connected in $G \cdot e$ then it implies that $G[S]$ is connected in $G$ as well.  Otherwise, by the induction hypothesis, in $G \cdot e$, there exists $u$ and $v$ with the desired property.  In particular, $V(X') \cap N_{G \cdot e}(u)$ and $V(X') \cap N_{G \cdot e}(v)$ are non empty where $X'=X \setminus \{x,y\} \cup \{z_{xy}\}$ and $X$ is the connected component in $S$ containing $x$ and $y$.  Thus, both $u$ and $v$ have the desired property in $G$ too.  \\
{\bf Case 2:} $G[S]$ is an independent set.  Now consider $x,y \in S$.  Consider the graph $G \cdot xy$ obtained by contracting the non-adjacent pair $\{x,y\}$.  Let $S'=S \setminus \{x,y\} \cup \{z_{xy}\}$ and edges incident on $x$ or $y$ are now incident on $z_{xy}$.  Observe that $S'$ is a minimal $(s,t)$-vertex separator in $G \cdot xy$.  Clearly, $|V(G \cdot xy)|=|V(G)|-1$ and hence, by the induction hypothesis, in $G \cdot xy$, there exists $u$ in $C'_s$ and $v$ in $C'_t$ satisfying our claim where $C'_s$ and $C'_t$ are connected components in $(G \cdot xy) \setminus S'$ containing $s$ and $t$, respectively.  Let $S=\{x,y,u_1,\ldots,u_p\},p \geq 1$.  We now prove in $G$ the existence of vertex $u$ in $C_s$ satisfying our claim.  If $\{u,x\},\{u,y\} \in E(G)$, then clearly $u \in C_s$ is the desired vertex in $G$.  Otherwise, without loss of generality assume that $x \notin N_G(u)$. Thus, $S \setminus \{x\} \subset N_G(u)$.  Let $P^s_{xu}$ denote a shortest path between $x$ and $u$ such that the internal vertices are in $C_s$.  Consider the vertex $w$ in $P^s_{xu}$ such that $\{x,w\} \in E(G)$.  Such a $w$ exists as $S$ is a minimal $(s,t)$-vertex separator in $G$.  If for all $z \in S$, $\{w,z\} \in E(G)$, then $w$ is a desired vertex in $C_s$.  Otherwise, there exists $z \in S$ such that $\{w,z\} \notin E(G)$.   Let $P^s_{wu}$ denote the subpath of $P^s_{xu}$ on the vertex set $\{w=w_1,\ldots,w_q=u\},q \geq 2$.  If for each $2 \leq i \leq q-1$, $\{z,w_i\} \notin E(G)$, then $P^s_{xu}\{u,z\}P^t_{xz}$ form an induced cycle of length at least 5 in $G$ where $P^t_{xz}$ denote a shortest path between $x$ and $z$ such that the internal vertices are in $C_t$.  Otherwise, for some $2 \leq i \leq q-1$, $\{z,w_i\} \in E(G)$.  Let $i$ be the smallest integer such that $\{z,w_i\} \in E(G)$.  In this case, $P^s_{xw_i}\{w_i,z\}P^t_{xz}$ form an induced cycle of length at least 5 in $G$ where $P^s_{xw_i}$ denote the subpath of $P^s_{xu}$ on the vertex set $\{x,w=w_1,\ldots,w_i\}, 2 \leq i \leq q-1$.  However, this contradicts the fact that $G$ is a graph of chordality 4.  Therefore, there exists a vertex $\hat{u}\in\{u,w\}$ in $C_s$ such that $S \subseteq N_G(\hat{u})$.  The proof for the existence of vertex $v$ in $C_t$ such that $S \subseteq N_G(v)$ is symmetric. 
\qed
\end{proof}
\begin{lemma}
\label{stcvs-ch4-lemma}
Let $G$ be a chordality 4 graph with the $(s,t)$-vertex connectivity $k$.  The size of any minimum $(s,t)$-CVS in $G$ is either $k$ or $k+1$.  
\end{lemma}
\begin{proof}
Note that any minimum $(s,t)$-CVS is of size at least $k$ as the $(s,t)$-vertex connectivity is $k$.  If a minimum $(s,t)$-vertex separator itself is connected then we get a minimum $(s,t)$-CVS of size $k$.  Otherwise, there exists a minimum $(s,t)$-vertex separator $S$ such that $G[S]$ is a collection of connected components.  In this case, we know from Theorem \ref{ch4-char-thm}, there exists a vertex $v$ in one of the components of $G \setminus S$ such that $S \subseteq N_G(v)$.   Therefore, $S \cup \{v\}$ is a minimum $(s,t)$-CVS of size $k+1$.  Hence the claim.  \qed
\end{proof}
Using the above two claims, and with the help of the following two key combinatorial observations on the structure of minimal vertex separators in chordality 4 graphs we show that $(s,t)$-CVS in chordality 4 graphs is polynomial-time solvable.  We make use of the notion of {\em contractible edges}.  Given a connected graph $G$ with $(s,t)$-vertex connectivity $k$, an edge $e \in E(G)$ is said to be {\em contractible} if $(s,t)$-vertex connectivity in $G \cdot e$ is at least $k$.  Otherwise $e$ is called {\em non-contractible}.  
\begin{lemma}
\label{obse}
Let $G$ be a connected graph and $S$ be a minimum $(s,t)$-vertex separator with $|S| \geq 2$. Let $F =\{\{u,v\} ~|~ u,v \in S$ and $\{u,v\} \in E(G)\}$.  $G[S]$ is connected if and only if the $(s,t)$-vertex connectivity in $G \cdot F$ (the graph obtained by contracting $F$ in $G$) is one.   Further, $G \cdot F$ contains a cut-vertex.
\end{lemma}
\begin{proof}
Clearly each edge contained in $S$ is non-contractible and by contracting all non-contractible edges, $S$ becomes a cut-vertex.  Moreover, any edge contraction does not disconnect a graph which is already connected.  Therefore, $\kappa(G \cdot F)=1$.  We prove the converse by contradiction.  Suppose, for all minimum $(s,t)$-vertex separator $S$, $G[S]$ contains at least two components.  This implies that the graph resulting from contracting any sequence of non-contractible edges has $(s,t)$-vertex connectivity at least two.  A contradiction to the fact $G \cdot F$ contains a cut-vertex.  Hence the claim. \qed
\end{proof}
Using the above lemma, we can decide in polynomial time whether a chordality 4 graph with the $(s,t)$-vertex connectivity $k$ contains a $(s,t)$-CVS of size $k$ or $k+1$.  The approach is to contract all non-contractible edges and check whether the resulting graph contains a cut-vertex or not.  If so, then the given chordality 4 graph contains a $(s,t)$-CVS of size $k$.  Otherwise, any minimum $(s,t)$-vertex separator in $G$ together with the vertex $v$ in one of the components in $G \setminus S$ (due to Theorem \ref{ch4-char-thm}) yields a $(s,t)$-CVS of size $k+1$ in $G$.   Our next combinatorial observation help us in finding a minimum $(s,t)$-CVS in polynomial time.  
\begin{lemma}
Let $G$ be a chordality 4 graph with the $(s,t)$-vertex connectivity $k$ and $e=\{u,v\}$ be a non-contractible edge contained in minimum vertex separators $S$ and $S'$ such that $G[S]$ is connected and $G[S']$ is not connected.  Then, for each $x \in S' \setminus S$, there exists $y \in S$, such that $\{x,y\} \in E(G)$.   
\end{lemma}
\begin{proof}
Clearly $k \geq 3$.  Let $S=\{x_1=u,x_2=v,\ldots,x_k\}$, $S'=\{y_1=u,y_2=v,\ldots,y_k\}$, and $S'\setminus S=\{y_i,\ldots,y_j\}, i \geq 3,j \leq k$. We prove by induction on $\kappa(G)$.  For the base case, $\kappa(G)=3$. i.e., $|S' \setminus S|=1$.   Suppose $y_k \in S'\setminus S$ is such that for any $x_i \in S$, $\{x_i,y_k\} \notin E(G)$.  In particular, $\{x_3,y_3\} \notin E(G)$.   This implies that any shortest path $P_{x_3y_3}$ between $x_3$ and $y_3$ is of length at least 2 in $G$.  Since $S$ is a minimal vertex separator, there exists $w$ in $C_1$ such that $\{x_2,w\} \in E(G)$, where $C_1$ is a connected component in $G \setminus S$.  Note that, the paths $P_{x_3y_3}, P_{wy_3}$, and $P'_{x_2y_3}$ induces a cycle of length at least 5 in $G$, where $P'_{x_2y_3}$ is a shortest path whose internal vertices are in $S$ or in $C_2$, where $C_2 (\neq C_1)$ is a connected component in $G \setminus S$.  Observe that there can not be chords from $w$ to any internal vertex in $P_{x_3y_3}$ as $S'$ is a vertex separator in $G$.  Hence, our assumption that $\{x_3y_3\} \notin E(G)$ is wrong.  Therefore, the claim is true for the base case.  If suppose there exists a chord between $x_2$ and some internal vertex, say $z$ in $P_{x_3y_3}$, then $\{x_1,x_2,z\}$ is a connected minmum vertex separator and we run the above argument by considering $S=\{x_1,x_2,z\}$.   We assume that our claim is true for all connected graphs with $\kappa(G) < k, k\geq 3$ and satisfying the premise of the lemma.  Let $G$ be a chordality 4 graph with $\kappa(G)=k, k \geq 4$.  Consider $S$ and $S'$ as defined in the premise of the lemma.   Consider the graph $G'$ obtained from $G$ by contracting the pairs $\{x_k,x_{k-1}\}$ and $\{y_j,y_{j-1}\}$ and $p$ and $q$ are the newly created vertices due to contraction, respectively.  Clearly, $\kappa(G')=k-1$ as $S$ and $S'$ are minimum vertex separators of size $k-1$.  By our induction hypothesis in $G'$, for each $x \in S' \setminus S$, there exists $y \in S$, such that $\{x,y\} \in E(G')$.  Since $q$ has a neighbour in $S$ in $G'$, at least one of $y_j$ or $y_{j-1}$ must have a neighbour in $S$ in $G$.  If suppose in $G$, $y_{j-1}$ has a neighbour in $S$ and $y_j$ does not have a neighbour in $S$.  An argument similar to the base case produces an induced cycle of length at least 5 containing $y_j$.  This completes the induction and therefore, the lemma follows. \qed
\end{proof}
\begin{algorithm}
\begin{enumerate}
\item[]{\bf A polynomial-time algorithm to find $(s,t)$-CVS in chordality 4 graphs} \\ 
\item[]{\bf Input: Chordality 4 graph with the $(s,t)$-vertex connectivity $k$}
\item[1.] Find the set $E_c$ of all contractible edges in $G$.
\item[2.] Contract $E_c$ and check the $(s,t)$-vertex connectivity in the resulting graph $G_c$.
\item[3.] If $\kappa(G_c)=1$, then the size of minimum $(s,t)$-CVS is $k$.  Otherwise it is $k+1$.
\item[4.] If $\kappa(G_c) \geq 2$, then any minimum vertex separator $S$ in $G$ augmented with the vertex $v$ in one of the components in $G \setminus S$ such that $S \subseteq N_G(v)$ is a minimum $(s,t)$-CVS of size $k+1$.
\item[5.] If $\kappa(G_c)=1$, then there exists a connected minimum vertex separator.  To obtain one such separator, perform the following;
\item[6] For each non-contractible edge $e$ in $G$, contract $e$ and find a minimum vertex separator $S'$ in $G \cdot e$.
\begin{itemize}
\item[6a.] If $G[S']$ is connected, then output the minimum vertex separator $S$ in $G$ corresponding to $S'$.
\item[6b.] If $G[S']$ is not connected, then check whether there exists a connected minimum vertex separator $S''$ in the neighbourhood of $S'$.
\end{itemize}
\end{enumerate}
\end{algorithm}
\subsection{$(s,t)$-CVS Parameterized by treewidth is Fixed-parameter Tractable}
\label{mso}
We transform the instance of $(s,t)$-CVS problem to the satisfiability of a formula in {\em monadic second order logic} (MSOL).  Using Courcelle's Theorem \cite{courcelle} that a problem over bounded treewidth graphs expressible in MSOL can be solved in linear time.   A recent paper by Marx \cite{marx} uses the same approach to prove the fixed-parameter tractability of other constrained separator problems.
We now present the description of monadic second order logic for $(s,t)$-CVS.  The atomic predicates used are as follows:  For a set $S \subseteq V(G)$, $S(v)$ denotes that $v$ is an element of $S$,  and the predicate $E(u,v)$ denotes the adjacency between $u$ and $v$ in $G$. $T=\{s,t\}$ and $k$ is the upper bound on the size of the desired $(s,t)$-CVS.  We construct the formula $\phi$ in MSOL as
\begin{center}
$\phi=\exists S(AtMost_k(S) \wedge Separates(S) \wedge ConnectedSubgraph(S))$
\end{center}
Here the predicate $AtMost_k(S)$ is true if and only if $|S| \leq k$, $Separates(S)$ is true if and only if $S$ separates the vertices of $T$  in $G$, and $ConnectedSubgraph(S)$ is true if and only if $S$ induces a connected subgraph in $G$. 
We refer \cite{marx} for formulae $AtMost_k(S)$  and $Separates(S)$.  The formula $Connects(Z,s,t)$ is due to \cite{gottlob}, where
$Connects(Z,s,t)$ is true if and only if in $G$, there is a path from $s$ and $t$ to all vertices of which belong to $Z$.  $ConnectedSubgraph(S)$ is true if and only if for every subset $S'$ of $S$ there is an edge between $S'$ and $S \setminus S'$.
\begin{itemize}
\item $AtMost_k(S): \forall c_1,\ldots,\forall c_{k+1} {\displaystyle \bigvee_{1 \leq i,j \leq k+1}}(c_i=c_j)$ 
\item $Separates(S): \forall s \forall t \forall Z (T(s) \wedge T(t) \wedge \neg(s=t) \wedge \neg S(s) \wedge \neg S(t) \wedge Connects(Z,s,t)) \rightarrow (\exists v(S(v) \wedge Z(v))))$,  $Connects(Z,s,t): Z(s)\wedge Z(t) \wedge \forall P((P(s) \wedge \neg P(t)) \rightarrow (\exists v \exists w (Z(v) \wedge Z(w) \wedge P(v) \wedge \neg P(w) \wedge E(v,w))))$ 
\item $ConnectedSubgraph(S):  \forall S' \subseteq S((S \not=S') \wedge (\exists u(S(u) \wedge S'(u))) \rightarrow (\exists u \exists w(S(u) \wedge S(w) \wedge (u,w) \in E(G) \wedge S'(v) \wedge \neg S'(w))) $ 
\end{itemize}
This complete the observation that in bounded treewidth graphs, a minimum $(s,t)$-CVS can be found in linear time. \\
{\bf Concluding Remarks and Further Research:}\\
In this paper, we have investigated the complexity of connected $(s,t)$-vertex separator ($(s,t)$-CVS) and shown that for every $\epsilon >0$, $(s,t)$-CVS is $\Omega(log^{2-\epsilon}n)$-hard, unless NP has quasi-polynomial Las-Vegas algorithms.  Also shown that $(s,t)$-CVS is NP-complete on graphs with chordality at least 5 and presented a polynomial-time algorithm for $(s,t)$-CVS on chordality 4 graphs.  Moreover, parameterizing above $(s,t)$-vertex connectivity is $W[2]$-hard.  An interesting problem for further research is to parameterize  $(s,t)$-CVS by the $(s,t)$-vertex connectivity.
\bibliographystyle{splncs}
\bibliography{stacs-connected-separator}
\end{document}